\documentclass[envcountsame]{llncs}
\usepackage[T1]{fontenc}
\usepackage{ltl}
\usepackage{comment}

\usepackage{xspace}
\usepackage{xcolor,stmaryrd}
\usepackage{scalerel,multicol}
\usepackage{amsmath}
\usepackage{mathtools}            
\usepackage{tabularx}
\usepackage{multirow}
\usepackage{makecell}
\usepackage{thmtools, comment, url} 
\usepackage{ifthen}
\usepackage{booktabs, todonotes}
\usepackage{hyperref}
\usepackage{multicol}
\usepackage{microtype}
\usepackage[misc,geometry]{ifsym}

\usepackage[mathlines]{lineno}
\usepackage{etoolbox}

\DeclareMathOperator\dep{\mathrm{dep}}

\newcommand{\LTL}{\protect\ensuremath{\mathrm{LTL}}\xspace}
\newcommand{\CTL}{\protect\ensuremath{\mathrm{CTL}}\xspace}
\newcommand{\QPTL}{\protect\ensuremath{\mathrm{QPTL}}\xspace}

\newcommand{\teamltl}{\protect\ensuremath{\mathrm{TeamLTL}}\xspace}

\newcommand{\hyctl}{\protect\ensuremath{\mathrm{HyperCTL^*}}\xspace}
\newcommand{\hyltl}{\protect\ensuremath{\mathrm{HyperLTL}}\xspace}
\newcommand{\ovuhyltl}{\protect\ensuremath{\forall\mathrm{HyperLTL}}\xspace}

\newcommand{\HQPTLP}{\protect\ensuremath{\mathrm{HyperQPTL\textsuperscript{\hskip-2pt\small +}}}\xspace}
\newcommand{\HQPTL}{\protect\ensuremath{\mathrm{HyperQPTL}}\xspace}
\newcommand{\FOE}{\protect\ensuremath{\mathrm{FO[\leq,E]}}\xspace}
\newcommand{\FOo}{\protect\ensuremath{\mathrm{FO[\leq]}}\xspace}

\newcommand{\dfn}{\mathrel{\mathop:}=}
\newcommand{\ddfn}{\mathrel{\mathop{{\mathop:}{\mathop:}}}=}
\newcommand{\lmodels}{\mathrel{\models^l}}

\newcommand{\pow}[1]{2^{#1}}

\newcommand{\BC}[1]{\mathrm{BC}( #1 )}
\newcommand{\PBC}[1]{\mathrm{PBC}( #1 )}

\newcommand{\U}{\LTLuntil}

\newcommand{\X}{\LTLnext}
\newcommand{\G}{\LTLg}
\newcommand{\F}{\LTLf}

\newcommand{\N}{\mathbb N}
\newcommand{\ap}{\mathrm{AP}}

\newcommand{\NP}{\protect\ensuremath{\mathrm{NP}}}

\newcommand{\PL}{\protect\ensuremath{\mathrm{PL}}}
\newcommand{\Ptime}{\protect\ensuremath{\mathrm{P}}}

\renewcommand{\phi}{\varphi}
\renewcommand{\epsilon}{\varepsilon}

\newcommand{\clor}{\varovee}

\DeclareMathOperator*{\Clor}{\scalerel*{\ovee}{\sum}}

\newcommand{\cneg}{\sim}

\newcommand{\teamctl}{\protect\ensuremath{\mathrm{Team}\CTL}}

\date{\today}

\begin{document}

\title{A Remark on the Expressivity of Asynchronous 
    TeamLTL and HyperLTL}
\author{Juha Kontinen\inst{1}\orcidID{0000-0003-0115-5154}, 
    Max Sandström\Letter\inst{1,2}\orcidID{0000-0002-6365-2562}, 
    Jonni Virtema\inst{1,2}\orcidID{0000-0002-1582-3718}
    }
\authorrunning{J. Kontinen et al.}
\institute{Department of Mathematics and Statistics, 
    University of Helsinki,
    Helsinki,
    Finland
    \email{\{juha.kontinen, max.sandstrom\}@helsinki.fi}
    \and
    Department of Computer Science, 
    University of Sheffield, 
    Sheffield,
    UK
    \email{j.t.virtema@sheffield.ac.uk}}

\maketitle

\begin{abstract}
Linear temporal logic (LTL) is used in system verification
to write formal specifications for reactive systems.
However, some relevant properties, 
e.g. non-inference in information flow security,
cannot be expressed in LTL.
A class of such properties that has recently 
received ample attention 
is known as hyperproperties.
There are two major streams in the research regarding 
capturing hyperproperties, namely
hyperlogics, which extend LTL with trace quantifiers (HyperLTL),
and logics that employ team semantics, 
extending truth to sets of traces.
In this article we explore the relation between asynchronous 
LTL under set-based 
team semantics (TeamLTL) and HyperLTL.
In particular we consider the extensions of TeamLTL with
the Boolean disjunction and a fragment of the extension of
TeamLTL with the Boolean negation, where the negation
cannot occur in the left-hand side of the Until-operator or 
within the Global-operator.
We show that TeamLTL extended with 
the Boolean disjunction is equi-expressive with
the positive Boolean closure of HyperLTL restricted to 
one universal quantifier,
while the left-downward closed fragment of 
TeamLTL extended with the Boolean negation
is expressively equivalent with 
the Boolean closure of HyperLTL restricted to 
one universal quantifier.

\end{abstract}

\section{Introduction}
In 1977 Amir Pnueli~\cite{pnueli} introduced
a core concept in verification of reactive and 
concurrent systems: model checking of formulae of
linear temporal logic (\LTL). 
The idea is to view the accepting executions of 
the system as a set of infinite sequences, called traces,
and check whether this set satisfies specifications
expressed in \LTL.
The properties that can be checked by observing
every execution of the system in isolation are called
\emph{trace properties}.
An oft-cited example of a trace property is \emph{termination},
which states that a system terminates if each of its computations
terminates.
Classical \LTL is fit for 
the verification of such propositional trace properties,
however some properties relevant in, for instance, 
information flow security are not trace properties.
These properties profoundly speak of relations between traces.
Clarkson and Schneider \cite{clarkson} coined the term 
\emph{hyperproperties} to refer to such properties 
that lie beyond what \LTL can express.
\emph{Bounded termination} is an easy to grasp example of
a hyperproperty: 
whether every computation of 
a system terminates within some bound common for all traces, 
cannot be determined by looking at traces in isolation.
In information flow security, dependencies between 
public observable outputs and
secret inputs constitute possible security breaches; 
checking for hyperproperties
becomes invaluable.
Two well-known examples of hyperproperties from this field are
noninterference~
\cite{DBLP:conf/sp/Roscoe95,DBLP:journals/jcs/McLean92},
where a high-level user cannot affect 
what low-level users see, 
and observational determinism~\cite{DBLP:conf/csfw/ZdancewicM03},
meaning that if two computations are in the same state 
according to a low-level observer,
then the executions will be indistinguishable.
However, hyperproperties are not 
limited to information flow security;
examples from different fields 
include distributivity and 
other system properties such as 
fault tolerance~\cite{DBLP:journals/acta/FinkbeinerHLST20}.

Given this background, several approaches to 
formally specifying hyperproperties
have been proposed since 2010, 
with families of logics emerging from these approaches.
The two major streams in the research regarding 
capturing hyperproperties are
\emph{hyperlogics} and logics that employ \emph{team semantics}.
In the hyperlogics approach, 
logics that capture trace properties 
are extended with trace quantification, 
extending logics such as \LTL, computation tree logic (\CTL) or 
quantified propositional temporal logic (\QPTL), 
into \hyltl~\cite{DBLP:conf/post/ClarksonFKMRS14}, 
\hyctl~\cite{DBLP:conf/post/ClarksonFKMRS14}, 
and \HQPTL~\cite{MarkusThesis,DBLP:conf/lics/CoenenFHH19}, 
respectively.
An alternative approach is to lift the semantics
of the temporal logics from being 
defined on traces to sets of traces,
by using what is known as team semantics.
This approach yields logics such as 
\teamltl \cite{kmvz18,GMOV22} and
\teamctl\cite{KrebsMV15,GMOV22}.
Since its conception, \teamltl has been considered in 
two distinct variants:
a synchronous semantics, 
where the team of traces agrees on the time step
of occurrence when evaluating temporal operators;
and an asynchronous semantics, 
where the temporal operators are evaluated
independently on each trace.
An example that illustrates the difference between 
these two semantics is
the aforementioned termination and 
bound termination pair of properties.
If we write $\F$ for the future-operator and 
$\mathrm{terminate}$ for
a proposition symbol representing the trace terminating, 
we can write the formula $\F \mathrm{terminate}$,
which under the synchronous semantics expresses the hyperproperty
``bounded termination'',
while under the asynchronous semantics 
the same formula defines the trace property
``termination''.
Not only is the above formulation of 
bounded termination clear and concise,
it also illuminates a key difference between 
hyperlogics and team logics:
while each formula of hyperlogic has 
a fixed number of quantifiers, 
which restricts the number of traces that can be referred to
in a formula, which restricts the number of traces
between which dependencies can be characterised by formulae,
team logics have the ability to refer to 
an unbounded number of traces, even an infinite collection.

One of the original motivations behind 
team semantics \cite{vaananen07}
was to enable the definition of novel atomic formulae,
and this is another important defining feature of 
team temporal logics
as well.
Among these atoms the \emph{dependence atom} 
$\dep(\bar x,\bar y)$ and 
\emph{inclusion atom} 
$\bar x \subseteq \bar y$ 
stand out as the most influential. 
They respectively state that 
the variables $\bar y$ are 
functionally dependent on
the variables $\bar x$, 
and that the values of 
the variables $\bar x$
also occur among the values of variables $\bar y$.
As an example of the use of the inclusion atom, 
let the proposition symbols $o_1,\dots, o_n$ 
denote public observable bits and 
assume that the proposition symbol $s$ 
is a secret bit. 
The atomic formula 
$(o_1,\dots o_n, s) \subseteq (o_1,\dots o_n, \neg s)$ 
expresses a form of non-inference by stating that an observer 
cannot infer the value of the confidential bit from the outputs.

While the expressivity of \hyltl and 
other hyperlogics has been studied extensively, 
where the many extensions of \teamltl 
lie in relation the hyperlogics is still not 
completely understood. 
The connections for the logics without extensions were already
established in Krebs et al.~\cite{kmvz18}, 
where they showed that synchronous \teamltl and 
\hyltl are expressively incomparable and that 
the asynchronous variant collapses to $\LTL$.
With regards to the expressivity of synchronous semantics, 
Virtema et al.~\cite{VBHKF20} 
showed that the extensions of \teamltl 
can be translated to \HQPTLP, 
which in turn extends \hyltl with (non-uniform) 
quantification of propositions. 
Relating the logics to the first-order context, 
Kontinen and Sandstr\"om~\cite{KS21} defined 
Kamp-style translations from 
extensions of both semantics of \teamltl to 
the three-variable fragment of first-order team logic. 
It is worth noting that recently asynchronous hyperlogics 
have been considered  also in several other articles 
(see, e.g.,  \cite{GutsfeldMO21,BaumeisterCBFS21}).
An example of the significant rift between 
asynchronous and synchronous $\teamltl$ 
is that the asynchronous semantics is essentially
a first-order logic, 
while the synchronous semantics has 
second-order aspects.
Especially the set-based variant of asynchronous $\teamltl$
can be translated, using techniques in \cite{KS21}, 
into first-order logic under team semantics, 
which is known to be first-order logic \cite{vaananen07}.
Similarly, \hyltl is equally expressive as 
the guarded fragment of 
first-order logic with the equal level predicate,
as was shown by Finkbeiner and Zimmermann~\cite{Finkbeiner017}.

In this article we focus on exploring 
the connections between fragments of \hyltl and 
extensions of \teamltl.
The set-based asynchronous semantics that we consider here 
was defined in Kontinen et al.~\cite{KontinenSV23} in order to 
further study the complexity of 
the model checking problem for these logics.
Prior to that, the literature on temporal team semantics 
employed a semantics based on multisets of traces. 
In the wider team semantics literature, 
this often carries the name \emph{strict semantics}, 
in contrast to \emph{lax semantics} which is de facto a 
set-based semantics.
This relaxation of the semantics enabled 
the definition of normal forms for the logics, 
which we use in this article to explore the connection 
with \hyltl.

\textbf{Our contribution.}
We show correspondences in expressivity between 
the set-based variant of 
linear temporal logic under asynchronous team semantics and
fragments of the Boolean closure of \hyltl. 
In particular we show that \LTL under 
team semantics with the Boolean disjunction, 
$\teamltl(\clor)$, is equi-expressive with
the positive Boolean closure of \hyltl restricted to 
only one universal quantifier,
while the left downward closed fragment of $\teamltl(\sim)$
is equi-expressive with 
the Boolean closure of \hyltl restricted to 
one universal quantifier. 

\section{Preliminaries}
We begin by defining 
the variant of $\teamltl$ and 
its extensions,
as in \cite{KontinenSV23}.

Let $\ap$ be a set of \emph{atomic propositions}. 
The formulae of \LTL (over $\ap$) is attained by 
the grammar:
\[
	\varphi \ddfn p \mid
				 \neg p  \mid
				 \varphi \lor \varphi \mid
				 \varphi\land \varphi \mid 
				 \X \varphi \mid 
				 \G \varphi \mid 
				 \varphi \U \varphi , 
\]
where $p \in \ap$.
We follow the convention that 
all formulae of \teamltl are given in 
negation normal form, 
where $\neg$ is only allowed before atomic propositions,
as is customary when dealing with team semantics. 

We will consider the extensions of \teamltl with 
the Boolean disjunction $\clor$, denoted $\teamltl(\clor)$,
and Boolean negation $\sim$, denoted $\teamltl(\sim)$.

A \emph{trace} $t$ over $\ap$ is an 
infinite sequence of sets of proposition symbols from 
$(\pow{\ap})^\omega$. 
Given a natural number $i\in\N$, we denote by 
$t[i]$ the $(i+1)$th element of $t$ and by 
$t[i,\infty]$ the suffix $(t[j])_{j\geq i}$ of $t$.
We call a set of traces a \emph{team}.

We write 
$\mathcal{P}(\mathbb{N})^+$ 
to denote 
$\mathcal{P}(\mathbb{N})\setminus \{\emptyset\}$. 
For a team $T\subseteq (\pow{\ap})^\omega$ a function 
$f\colon T\rightarrow \mathcal{P}(\mathbb{N})^+$, 
we set 
$T[f,\infty]\dfn\{t[s,\infty]\mid t\in T,s\in f(t)\}$.
For $T'\subseteq T$, $
f \colon T \to \mathcal{P}(\mathbb{N})^+$, 
and $f' \colon T' \to \mathcal{P}(\mathbb{N})^+$, 
we define that $f'< f$ if and only if
\begin{align*}
    \forall t \in T': 
	    &\min(f'(t)) \leq \min(f(t)) \text{ and, }\\
	    &\text{ if }\max(f(t))
        \text{ exists, }\max(f'(t)) < \max(f(t)).
\end{align*}

\begin{definition}[$\teamltl$]
Let $T$ be a team, and $\varphi$ and $\psi$ 
\teamltl-formulae. 
The lax semantics is defined as follows. 
    \begin{align*}
        & T\models l & 
            & \Leftrightarrow & 
            & t\models l \text{ for all } t\in T, 
	              \text{ where $l\in 
                \{p,\neg p \mid p\in \mathrm{AP}\}$} \\
        & & & & & \text{ is a literal and ``} t\models 
                \text{'' refers to \LTL-satisfaction} \\ 
        & T\models \varphi\wedge\psi &
            & \Leftrightarrow &
            & T\models\varphi \text{ and } T\models\psi \\
        & T \models\varphi\vee\psi & 
            & \Leftrightarrow & 
            & \exists T_1,T_2 \text{ s.t. }  
                T_1\cup T_2=T\text{ and }  
	            T_1\models\varphi\text{ and }
                T_2\models\psi \\
        & T\models \X\varphi & 
            & \Leftrightarrow & 
            & T[1,\infty]\models \varphi \\    
        & T\models \G\varphi & 
            & \Leftrightarrow & 
            & \forall f\colon T\rightarrow 
                \mathcal{P}(\mathbb{N})^+ 
	            \text{ it holds that } 
                T[f,\infty] \models\varphi \\  
    & T\models\varphi \U\psi & & \Leftrightarrow & & 
	   \exists f\colon T\rightarrow \mathcal{P}(\mathbb{N})^+  
	   \text{ such that } T[f,\infty] \models\psi\mbox{ and } \\
    & & & & & \forall f'\colon T'\rightarrow \mathcal{P}(\mathbb{N})^+ 
	   \text{s.t. $f' < f$, it holds that } 
	   T'[f',\infty] \models\varphi \\
    & & & & & \text{ or } T'= \emptyset, \text{ where }T'\dfn
      \{t \in T \mid \max(f(t))\neq 0  \} 
    \end{align*}
The semantics for the Boolean disjunction and Boolean negation, 
used in the extensions $\teamltl(\clor)$ and $\teamltl(\sim)$,
are given by:
    \begin{align*}
        & T\models \varphi\clor\psi & 
            & \Leftrightarrow & 
            & T\models \varphi \text{ or } T\models \psi \\
        & T\models \sim\varphi & 
            & \Leftrightarrow & 
            & T\not\models \varphi 
    \end{align*}
\end{definition}

Note that the Boolean disjunction is 
definable in $\teamltl(\sim)$,
as the dual of conjunction, 
i.e. $T \lmodels \varphi\clor\psi$ 
if and only if 
$T \lmodels \sim(\sim\varphi \wedge \sim\psi)$.

Two important properties of team logics are 
\emph{flatness} and \emph{downward closure}.
A logic has the flatness property if 
$T \lmodels \varphi$
if and only if 
$\{ t \} \lmodels \varphi$ for all $t\in T$,
holds for all formulae $\varphi$ of the logic.
A logic is downward closed if for
all formulae $\varphi$ of the logic
if $T\lmodels \varphi$ and $S \subseteq T$
then $S \lmodels \varphi$.
The following Proposition was proven in \cite{KontinenSV23}.
\begin{proposition}
    $\teamltl^l$ has both the flatness and 
    the downward closure properties, while 
    $\teamltl^l(\clor)$ only has the downward 
    closure property.
\end{proposition}

We consider the \emph{left-downward closed} fragment of
$\teamltl^l(\sim)$, denoted $\text{left-dc--}\teamltl^l(\sim)$,
where every subformula of the form $\G\psi$ or $\psi\U\theta$, 
the subformula $\psi$ is a $\teamltl(\clor)$-formula

It was established in \cite{KontinenSV23} that any formula of 
$\teamltl^l(\clor)$ can be equivalently expressed in 
\emph{$\clor$-disjunctive normal form}, i.e. in the form
\[
    \Clor_{i\in I} \alpha_i, 
\]
where $\alpha_i$ are \LTL-formulae.

Similarly by \cite{KontinenSV23}, 
every formula of $\text{left-dc--}\teamltl^l(\sim)$
can be equivalently stated in \emph{quasi-flat normal form},
which means in the form
\[ \Clor_{i\in I}(\alpha_i\wedge 
	\bigwedge_{j\in J_i} \exists \beta_{i,j}),  \]
where $\alpha_i$ and  $\beta_{i,j}$ are \LTL-formulae, 
and $\exists \beta_{i,j}$ is an abbreviation for 
the formula $\sim \beta^d_{i,j}$, 
where $\beta^d_{i,j}$ is 
the formula obtained from $\neg \beta$, 
after $\neg$ has been pushed down to the atomic level.

Next we state the syntax and semantics of \hyltl, 
as defined in \cite{DBLP:conf/post/ClarksonFKMRS14}, 
as well as the Boolean closure concepts we are concerned with.

\begin{definition}[Syntax of \hyltl]
Let $\ap$ be a set of propositional variables and 
$\mathcal{V}$ the set of all trace variables. 
Formulas of \hyltl are generated by the following grammar:
    \begin{alignat*}{10}
        \psi 
            &\ddfn& 
            &\exists\pi.\psi & 
            &\mid& 
            &\forall\pi.\psi &  
            &\mid& 
            &\varphi & 
            & & 
            & & 
            & & \\
        \varphi 
            &\ddfn& 
            &a_\pi&  
            &\mid& 
            &\neg\varphi &
            &\mid& 
            &\varphi\vee\varphi & 
            &\mid& 
            &\X\varphi &  
            &\mid& 
            &\varphi\U\varphi,
    \end{alignat*}
    where $a \in \ap$ and $\pi \in\mathcal{V}$. 
\end{definition}

We denote the set of all traces by $\mathrm{TR}$ and 
the set of all trace variables by $\mathcal{V}$. 
For a trace assignment function 
$\Pi\colon \mathcal{V} \to \mathrm{TR}$, 
we write $\Pi[i,\infty]$ for 
the trace assignment defined through 
$\Pi[i,\infty] = \Pi(\pi)[i,\infty]$, 
and $\Pi[\pi \mapsto t]$ for the assignment that 
assigns $t$ to $\pi$, 
but otherwise is identical to $\Pi$.

\begin{definition}[Semantics of \hyltl]
    Let $a \in \ap$ be a proposition symbol,
    $\pi \in \mathcal{V}$ be a trace variable,
    $T$ be a set of traces, and 
    let $\Pi\colon \mathcal{V} \to \mathrm{TR}$ be 
    a trace assignment.

    \begin{align*}
        &\Pi \models_T \exists\pi.\psi& 
            &\Leftrightarrow& 
            &\text{there exists } t\in T \colon 
                \Pi[\pi\mapsto t] \models_T \psi \\
        &\Pi \models_T \forall\pi.\psi& 
            &\Leftrightarrow& 
            &\text{for all } 
                t\in T \colon \Pi[\pi\mapsto t] \models_T \psi \\
        &\Pi \models_T a_\pi& 
            &\Leftrightarrow& 
            &a\in \Pi(\pi)[0] \\
        &\Pi \models_T \neg\varphi& 
            &\Leftrightarrow& 
            &\Pi\not\models_T \varphi \\
        &\Pi \models_T \varphi_1\vee\varphi_2& 
            &\Leftrightarrow& 
            &\Pi \models_T \varphi_1 \text{ or } 
                \Pi \models_T \varphi_2 \\
        &\Pi \models_T \X\varphi& 
            &\Leftrightarrow& 
            &\Pi[1,\infty] \models_T \varphi \\
        &\Pi \models_T \varphi_1\U\varphi_2& 
            &\Leftrightarrow& 
            &\text{there exists } i\geq 0\colon 
                \Pi[i,\infty] \models_T \varphi_2 \\ 
        & & 
            & & 
            &\quad\text{ and for all } 0\leq j < i 
                \text{ we have }
                \Pi[j,\infty] \models_T \varphi_1
    \end{align*}
\end{definition}

\begin{definition}[Universal Fragments]
    The \textit{universal fragment} of \hyltl, 
    denoted by $\forall^*\hyltl$, 
    is the fragment of \hyltl with 
    no existential quantification. 
    We write $\ovuhyltl$ for 
    the \textit{one variable universal fragment} of \hyltl,
    and $Q\hyltl$ for the \textit{one variable fragment} of \hyltl.
\end{definition}

\begin{definition}[(Positive) Boolean Closure]
    The \textit{Boolean closure} of a logic $\mathcal{L}$, 
    denoted by $\BC{\mathcal{L}}$, is 
    the extension of $\mathcal{L}$ that is 
    closed under $\wedge$, $\vee$ and $\neg$. 
    The \textit{positive Boolean closure} of 
    a logic $\mathcal{L}$, 
    denoted by $\PBC{\mathcal{L}}$, is 
    the extension of $\mathcal{L}$ that is closed under 
    $\wedge$ and $\vee$.
\end{definition}

The semantics for the Boolean closures are attained by relaxing 
the definition of conjunction $\wedge$, 
disjunction $\vee$, and $\neg$ to 
apply to any formula of the Boolean closure.

Using a suitable algorithm, 
all $\BC{\mathcal{L}}$-formulae can be 
equivalently expressed in disjunctive normal form, 
i.e. as a disjunction of conjunctions 
with possibly a negation in 
front of each formula of $\mathcal{L}$. 
Similarly, all $\PBC{\mathcal{L}}$-formulae 
can be equivalently expressed as 
\[
    \bigvee_{i\in I}\bigwedge_{j \in J} \varphi_{i,j}
\]
for some formulae $\varphi_{i,j} \in \mathcal{L}$ and 
index sets $I$ and $J$. 
From here on we use $I$ and $J$ to denote arbitrary index sets.

\section{Correspondence between $\teamltl$ and \hyltl}

In this section we will explore 
the relationship between the logics by 
proving some correspondence theorems. 
First, however, we prove some pertinent propositions regarding 
the Boolean closure of \hyltl,
showing that conjunction, disjunction and 
negation distribute over
the quantifiers in a manner analogous to first-order logic.
We go through these propositions in some detail as,
although they appear familiar from the first-order setting,
\hyltl is usually considered only in the prenex normal form
and thus these basic results are not explicitly addressed in
the literature.
Moreover, the proofs feature arguments that
will be useful in subsequent proofs.


As usual, for logics $\mathcal{L}$ and $\mathcal{L}'$, 
we write $\mathcal{L}\leq \mathcal{L}'$, 
if for every $\mathcal{L}$-formula 
there exists an equivalent $\mathcal{L}'$-formula. 
We write $\mathcal{L}\equiv \mathcal{L}'$, 
if both $\mathcal{L}\leq \mathcal{L}'$ and 
$\mathcal{L}'\leq \mathcal{L}$.
\begin{proposition}\label{prop:pbcforallhyltl=hyltl}
    $\PBC{\forall^*\hyltl} \equiv \forall^*\hyltl$
\end{proposition}
\begin{proof}
    Let $\bigvee_{i \in I} \bigwedge_{j \in J} \psi_{i,j}$ 
    be an arbitrary formula of 
    $\PBC{\forall^*\hyltl}$. 
    If all $\psi_{i,j}$ are quantifier free, 
    we are done, as then 
    $\bigvee_{i \in I} \bigwedge_{j \in J} \psi_{i,j}$ is a 
    $\forall^*\hyltl$-formula. 
    Thus, we may assume that 
    $\psi_{i,j} = \forall\pi_1 \cdots 
        \forall \pi_n\varphi_{i,j}$ 
    for some $\LTL$-formula $\varphi_{i,j}$.
    Suppose 
    \[
        \Pi \models_T \bigvee_{i \in I} 
            \bigwedge_{j \in J} \forall \pi_1 
            \cdots \forall \pi_n\varphi_{i,j}.
    \] 
    Without loss of generality, 
    we may assume a uniform quantifier block in each conjunct, 
    as one can rename variables and take the largest quantifier
    block as the common one, 
    since redundant quantifiers do not effect evaluation.
    The previous is therefore equivalent with
    \[
        \Pi \models_T \bigvee_{i \in I} \forall \pi_1 \cdots 
            \forall \pi_n \bigwedge_{j \in J} \varphi_{i,j}.
    \]
    At this point, 
    we wish to push the disjunction past the quantifier block,
    but the variables would become entangled and 
    different traces could satisfy different disjuncts.
    We need to distinguish the variables of 
    the disjuncts from each other, 
    so we rename the trace quantifiers. 
    The previous evaluation is therefore equivalent with
    \[
        \Pi \models_T \forall \pi_1^1 \cdots 
            \forall \pi_1^i \cdots
            \forall \pi_n^1 \cdots \forall \pi_n^i 
            \bigvee_{i \in I} \bigwedge_{j \in J} 
            \varphi_{i,j}(\pi_1^1,\cdots,\pi_n^i).
    \]
    This is a formula of $\forall^*\hyltl$. \qed
\end{proof}

The following remark, familiar from first-order logics,
can be proven with a straight-forward induction over 
the length of the quantifier block.

\begin{remark}\label{rem:neghyltl}
For \hyltl-formula $Q_1\pi_1 \cdots Q_n\pi_n \psi$ it holds that
    \[
        \neg Q_1\pi_1 \cdots Q_n\pi_n \psi \equiv 
            Q_1^-\pi_1 \cdots Q_n^-\pi_n \neg\psi,
    \]
    where for every index $i$, 
    $Q_i$ are quantifiers $\forall$ or $\exists$, 
    and $Q_i^-$ is $\exists$ if $Q_i$ is $\forall$ 
    and vice versa.
\end{remark}
%
%

\begin{proposition}
    $\BC{\hyltl} \equiv \hyltl$
\end{proposition}
\begin{proof}
    Consider a $\BC{\hyltl}$-formula 
    $\bigvee_{i \in I}\bigwedge_{j \in J} \varphi_{i,j}$ in 
    disjunctive normal form,
    with either $\varphi_{i,j} \in \hyltl$ or 
    $\varphi_{i,j} = \neg\psi_{i,j}$ 
    for some formula $\psi_{i,j}\in \hyltl$. 
    By Remark \ref{rem:neghyltl} $\neg\psi_{i,j} \equiv         
        Q_1^{i,j}\pi_1^{i,j} \cdots 
        Q_n^{i,j}\pi_n^{i,j} \theta_{i,j}$, 
    where $\theta_{i,j} \in \LTL$. 
    Thus we may assume that $\varphi_{i,j}$ 
    only appears positively. 
    By a similar argument to that of the proof of Proposition 
    \ref{prop:pbcforallhyltl=hyltl} we get the following:
    \[
        \bigvee_{i \in I} \bigwedge_{j \in J} 
            Q_{1}^{i,j}\pi_1^{i,j} \cdots 
            Q_{n}^{i,j}\pi_n^{i,j} \psi_{i,j}
            \equiv
            Q_{1}^{1,1} \pi_1^{1,1} \cdots 
            Q_{1}^{1,j} \pi_1^{1,j} \cdots 
            Q_{n}^{i,j} \pi_n^{i,j} 
            \bigvee_{i \in I}\bigwedge_{j \in J} \psi_{i,j}.
    \] \qed
\end{proof}

One last remark before we get to 
the core results of this article, 
this time relating quantifier-free \hyltl-formulae with
\LTL-formulae. 
The remark can again be proven by induction on 
the structure of the formula.

\begin{remark}\label{rem:hyperified}
    Let $T$ be a team, $\Pi$ be a trace assignment, 
    $\pi$ be a trace variable,
    $\varphi$ be a \LTL-formula, 
    and let $\varphi(\pi)$ be 
    the \hyltl formula identical to $\varphi$,
    except every proposition symbol $p$ is replaced by $p_\pi$.
    Suppose $\Pi(\pi) = t$ for some $t \in T$. 
    Now the following equivalence holds
    \[
        \Pi\models_T \varphi(\pi) \iff t \models \varphi.
    \]
\end{remark}

Using the above propositions we may now proceed with proving
our main results:
correspondence theorems between team logics and 
the Boolean closures of hyperlogics.


Note that $\teamltl$ has no separation between 
closed and open formulae, 
and has no features to encode trace assignments. 
Thus, when $\varphi$ is a formula of some 
team based logic $\mathcal{L}$ and 
$\psi$ is a formula of a hyper logic $\mathcal{L}'$ 
without free variables, 
we say that $\varphi$ and $\psi$ are equivalent, 
if the equivalence 
$T \models \varphi \Leftrightarrow \emptyset \models_T \psi$, 
holds for all sets of traces $T$.
The notations $\mathcal{L}\leq \mathcal{L}'$ and 
$\mathcal{L}\equiv \mathcal{L}'$ 
are then defined in the obvious way, 
by restricting $\mathcal{L}'$ to 
formulae without free variables.
\begin{theorem}\label{thm:teamltl=pbcovuhyltl}
    $\teamltl^l(\clor) \equiv \PBC{\ovuhyltl}$
\end{theorem}
\begin{proof}
    Let $T$ be an arbitrary team and $\varphi$ 
    an arbitrary $\teamltl^l(\clor)$-formula.
    By \cite[Theorem 10]{KontinenSV23}, 
    we may assume that $\varphi$ is in 
    the form $\Clor_{i\in I} \alpha_i$, 
    where $I$ in an index set and 
    $\alpha_i$ are \LTL-formulae.
    We let $\alpha_i(\pi)$ denote 
    the \hyltl-formulae obtained from $\alpha_i$, 
    by replacing every proposition symbol $p$ by 
    $p_\pi$. 
    We obtain the following chain of equivalences:
    \begin{align*}
        T \models \Clor_{i\in I} \alpha_i
        &\iff 
        \text{there is } i \in I \text{ such that } 
            T \models \alpha_i \\
        &\iff
        \text{there is } i \in I \text{ such that } 
            t \models \alpha_i \text{ for all } t\in T \\
        &\iff
        \text{there is } i \in I \text{ such that } 
            \emptyset \models_T \forall\pi \alpha_i(\pi) \\
        &\iff
        \emptyset \models_T \bigvee_{i\in I} \forall\pi\alpha_i(\pi),
    \end{align*}
    where 
    the first equivalence follows from the semantics of $\clor$,
    the second equivalence holds by the flatness of $\alpha_i$,
    the third equivalence is due to 
    the semantics of $\forall$ and 
    Remark \ref{rem:hyperified},
    and the final equivalence follows from 
    the semantics of $\vee$.

    For the converse direction, consider an arbitrary 
    $\PBC{\ovuhyltl}$-sentence $\psi$.
    As noted above, $\psi$ is equivalent to a sentence
    $\bigvee_{i \in I} \bigwedge_{j \in J} 
        \forall\pi \varphi_{i,j}(\pi)$,
    where $\varphi_{i,j}(\pi)$, for every pair $i$ and $j$, 
    is a \hyltl-formula with $\pi$ as the only free variable. 
    Now by an argument similar to the proof of 
    Proposition \ref{prop:pbcforallhyltl=hyltl},
    $\emptyset \models_T \bigvee_{i \in I} \bigwedge_{j \in J} 
        \forall\pi \varphi_{i,j}(\pi)$
    if and only if
    $\emptyset \models_T \bigvee_{i \in I} \forall\pi 
        \bigwedge_{j \in J} \varphi_{i,j}(\pi)$.
    Equivalently then by the definition of the semantics of 
    the disjunction, 
    we may fix $i' \in I$ such that 
    $\emptyset \models_T \forall\pi 
        \bigwedge_{j \in J} \varphi_{i',j}(\pi)$.
    By the definition of the universal quantifier then 
    we get that 
    the previous is equivalent with
    $\emptyset [\pi \mapsto t] \models_T 
        \bigwedge_{j \in J} \varphi_{i',j}(\pi)$ for all $t \in T$.
    Now by Remark \ref{rem:hyperified}, 
    the previous holds if and only if
    $t \models \bigwedge_{j \in J} \varphi_{i',j}
        \text{ for all } t\in T$,
    which is equivalent to 
    $T \lmodels \bigwedge_{j \in J} \varphi_{i',j}$,
    due to the flatness property of $\teamltl$.
    Finally, by the semantics of the Boolean disjunction,
    the previous is equivalent with
    $T \models \Clor_{i \in I} \bigwedge_{j \in J} 
        \varphi_{i,j}$. \qed
\end{proof}
As a corollary we get that $\teamltl^l(\clor)$ is subsumed by 
the universal fragment of \hyltl, 
which follows from Theorem \ref{thm:teamltl=pbcovuhyltl} and 
the observations made in the proof of 
Proposition \ref{prop:pbcforallhyltl=hyltl}.

\begin{corollary}\label{cor:teamltl(clor)leqUhyltl}
    $\teamltl^l(\clor) \leq \forall^*\hyltl$
\end{corollary}

Note that another consequence of 
Theorem \ref{thm:teamltl=pbcovuhyltl}
is that $\ovuhyltl$ is strictly 
less expressive than $\PBC{\ovuhyltl}$, 
as the former is equivalent with \LTL\cite{Finkbeiner17} 
and thus has the flatness property, 
while the latter is equivalent with $\teamltl(\clor)$,
which does not satisfy flatness. 
This stands in contrast to the unrestricted 
universal fragment $\forall^*\hyltl$, 
which by Proposition \ref{prop:pbcforallhyltl=hyltl} 
is equivalent to its positive Boolean closure.

\begin{theorem}\label{thm:teamltlsim=bcovuhyltl}
    $\text{left-dc--}\teamltl^l(\sim) 
    \equiv 
    \BC{Q\hyltl} \equiv \BC{\ovuhyltl}$
\end{theorem}
\begin{proof}
    Let $\varphi$ be an arbitrary 
    $\text{left-dc--}\teamltl^l(\sim)$-formula.
    Now by the quasi-flat normal form $T \lmodels \varphi$
    if and only if 
    $T \lmodels \Clor_{i \in I} \big( \alpha_i \wedge 
            \bigwedge_{j \in J} \exists \beta_{i,j} \big)$.
    Equivalently, by the semantics of $\clor$, 
    we may fix an index $i' \in I$ such that 
    $T \lmodels \alpha_{i'} \wedge 
            \bigwedge_{j \in J} \exists \beta_{i',j}$.
    By the semantics of the logic,
    flatness, and the interpretation of the shorthand $\exists$,
    the previous evaluation is equivalent with that
    $t \models \alpha_{i'}$ for all $t \in T$ and
    for all $j \in J$ there is a $t_j \in T$ such that
    $t_j \models \beta_{i',j}$.
    By Remark \ref{rem:hyperified} the previous holds 
    if and only if
    $\emptyset[\pi \mapsto t] \models_T \alpha_{i'}(\pi)$ 
    for all $t \in T$
    and for all $j \in J$ there is a $t_j \in T$ such that
    $\emptyset[\pi \mapsto t_j] \models_T \beta_{i',j}(\pi)$.
    Equivalently, by the semantics of $\forall$ and $\exists$,
    we have that
    $\emptyset \models_T \forall\pi \alpha_{i'}(\pi)$ and
    $\emptyset \models_T \bigwedge_{j \in J}
        \exists\pi \beta_{i',j}(\pi)$,
    which, finally by the semantics of $\vee$ and $\wedge$, 
    holds if and only if
    $\emptyset \models_T \bigvee_{i \in I} 
            \big( \forall\pi \alpha_i(\pi) \wedge 
            \bigwedge_{j \in J} 
            \exists\pi \beta_{i,j}(\pi)\big)$.

    On the other hand, 
    let $\psi$ be an arbitrary sentence of $\BC{\hyltl}$.
    Now we get the following chain of equivalences, 
    where $Q_{i,j}\in \{\exists, \forall\}$:
    {\allowdisplaybreaks
    \begin{align*}
        \emptyset \models_T \psi
        & \iff &
        & \emptyset \models_T \bigvee_{i \in I} \big(
            \bigwedge_{j \in J} Q_{i,j}\pi \varphi_{i,j} \big) \\
        & \iff & 
        & \emptyset \models_T \bigvee_{i \in I} \big(\forall\pi 
            \alpha_i \wedge \bigwedge_{j \in J} 
            \exists\pi \varphi_{i,j} \big) \\
        & \iff &
        & \text{there is } i \in I \text{ such that }
            \emptyset \models_T \forall\pi \alpha_i,
            \text{ and for all } j \in J \\
            & & &\text{it holds that } 
            \emptyset \models_T \exists\pi \varphi_{i,j}\\
        & \iff &
        & \text{there is } i \in I \text{ such that } 
            \emptyset[\pi \mapsto t] \models_T \alpha_i 
            \text{ for all } t\in T,
            \text{ and for all } \\
            & & & j \in J \text{ there exists } t_j \in T
            \text{ such that } 
            \emptyset[\pi \mapsto t_j] \models_T
            \varphi_{i,j} \\
        & \iff & 
        & \text{there is } i \in I \text{ such that }
            t \models \alpha_i \text{ for all } t\in T,
            \text{ and for all } j\in J \\
            & & & \text{there exists }
            t_j \in T \text{ such that } 
            t_j \models \varphi_{i,j} \displaybreak \\
        & \iff &
        & \text{there is } i \in I \text{ such that }
            T \models \alpha_i, \text{ and for all } j\in J
            \text{ it holds that } \\
            & & & T\models \exists \varphi_{i,j}\\
        & \iff &
        & \text{there is } i \in I \text{ such that }
            T \models \alpha_i \wedge \bigwedge_{j \in J} 
            \exists \varphi_{i,j} \\
        & \iff &
        & T \models \Clor_{i \in I} \big( \alpha_i \wedge 
            \bigwedge_{j \in J} \exists \varphi_{i,j} \big),
    \end{align*}
    }
    where the first equivalence is due to 
    the normal form for a Boolean closure,
    the second equivalence is holds because 
    the universally quantified conjuncts can 
    equivalently be evaluated simultaneously,
    the third equivalence follows from 
    the semantics of $\wedge$ and $\vee$,
    the fourth equivalence holds by 
    the semantics of $\forall$ and $\exists$,
    the fifth equivalence holds by 
    Remark \ref{rem:hyperified},
    the sixth equivalences is due to flatness and 
    the definition of the shorthand $\exists$,
    the seventh equivalence holds by 
    the semantics of $\wedge$, 
    and the final equivalence follows from 
    the semantics of $\clor$. 

    The other equivalence in the theorem is a 
    direct consequence of 
    Remark \ref{rem:neghyltl}.\qed
\end{proof}

\section{Conclusion}

In this article we explored 
the connections in expressivity between 
extensions of linear temporal logic under 
set-based team semantics ($\teamltl^l$)
and fragments of linear temporal logic 
extended with trace quantifiers (\hyltl).
We showed that $\teamltl^l$, 
when extended with the Boolean disjunction, 
corresponds to the positive Boolean closure of 
the one variable universal fragment of \hyltl. 
Furthermore we considered 
a fragment of $\teamltl^l$ extended with the Boolean negation,
where the formulae are restricted to 
not to contain the Boolean negation
on the left-hand side of the `until' operator ($\U$) or
under the `always going to' ($\G$) operator.
We showed a correspondence between that fragment and 
the Boolean closure of 
the one variable universal fragment of \hyltl.
From our results it follows that 
the logics considered are all 
true extensions of \LTL. 
Decidability of the model checking and satisfaction problems 
for the team based logics was shown in \cite{KontinenSV23},
and by our correspondence results
(and the translation implied by the proofs of the theorems), 
decidability of the problems 
extends to the hyperlogics in question as well.
It is fascinating to see that the restriction to
left downward closed formulae in 
the latter correspondence on the team logic side
disappears on the hyperlogic side.
This hints at that the fragment considered is intuitive.
It is an open question whether the downward closed
fragment of $\teamltl^l(\sim)$ is $\teamltl^l(\clor)$,
or if some restricted use of the Boolean negation
could be allowed and still maintain downward closure.
Another open question is whether 
an analogous correspondence exists for 
the full logic $\teamltl^l(\sim)$, 
or even for some lesser restriction of the logic 
than the left-downward closed fragment.

\begin{credits}
    \textbf{Acknowledgements.}
    This work was supported by 
    the Academy of Finland grant 345634 and 
    the DFG grant VI 1045/1-1.
\end{credits}
\bibliographystyle{plainurl}
\bibliography{biblio}

\end{document}